\documentclass[journal,compsoc]{IEEEtran}

\usepackage{ifpdf}
\usepackage{color}
\ifCLASSOPTIONcompsoc
  \usepackage{cite}
\else
\fi
\ifCLASSINFOpdf
  \usepackage[pdftex]{graphicx}
\else
  \usepackage[dvips]{graphicx}
\fi
\usepackage[cmex10]{amsmath}
\usepackage{amsfonts,amsmath}
\usepackage{cite}

\usepackage{graphicx}
\usepackage{amssymb}
\usepackage{amsthm}
\usepackage{mdwmath}
\usepackage{caption}
\usepackage{subcaption}
\usepackage{mdwtab}
\usepackage{enumerate}

\usepackage{algorithm, algpseudocode}
\usepackage{array}
\usepackage{eqparbox}
\usepackage{pifont}
\usepackage{epstopdf}

\theoremstyle{definition}

\newtheorem{thm}{Theorem}
\newtheorem{lem}{Lemma}

\newtheorem{cor}{Corollary}

\definecolor{Blue}{rgb}{0.0,0.0,1.0}
\ifCLASSINFOpdf
\else
\fi
\begin{document}
%
\title{PBF: A New Privacy-Aware Billing Framework for Online Electric Vehicles with Bidirectional Auditability}
%
%
%

\author{Rasheed Hussain,
        Donghyun Kim,~\IEEEmembership{Member,~IEEE,}
        Michele Nogueira,~\IEEEmembership{Member,~IEEE,}\\
        Junggab Son,
        Alade O. Tokuta,~\IEEEmembership{Member,~IEEE,}
        Heekuck Oh,~\IEEEmembership{Member,~IEEE}
\IEEEcompsocitemizethanks{
\IEEEcompsocthanksitem R. Hussain and H. Oh are with the Department
of Computer Science and Engineering, Hanyang University, South Korea. E-mail: \{rasheed,hkoh\}@hanyang.ac.kr.\protect
\IEEEcompsocthanksitem D. Kim, J.Son, and Alade O. Tokuta are with Department of Mathematics and Physics, North Carolina Central University, Durham, NC 27707, USA. E-mail: donghyun.kim@nccu.edu.\protect
\IEEEcompsocthanksitem M. Nogueira is with Department of Informatics, Federal University of Paran$\acute{\text{a}}$, Curitiba, Brazil. E-mail: michele@inf.ufpr.br.\protect
\IEEEcompsocthanksitem H. Oh is the corresponding author. \protect
}}

\IEEEcompsoctitleabstractindextext{%
\begin{abstract}

\textbf{Recently an online electric vehicle (OLEV) concept has been introduced, where vehicles are propelled through the wirelessly transmitted electrical power from the infrastructure installed under the road while moving. The absence of secure-and-fair billing is one main hurdle to widely adopt this promising technology. This paper introduces a secure and privacy-aware fair billing framework for OLEV on the move through the charging plates installed under the road. We first propose two extreme lightweight mutual authentication mechanisms, a direct authentication and a hash chain-based authentication between vehicles and the charging plates that can be used for different vehicular speeds on the road. Second we propose a secure and privacy-aware wireless power transfer on move for the vehicles with bidirectional auditability guarantee by leveraging game-theoretic approach. Each charging plate transfers a fixed amount of energy to the vehicle and bills the vehicle in a privacy-aware way accordingly. Our protocol guarantees secure, privacy-aware, and fair billing mechanism for the OLEVs while receiving electric power from the road. Moreover our proposed framework can play a vital role in eliminating the security and privacy challenges in the deployment of power transfer technology to the OLEVs.}
\end{abstract}
}

\maketitle

\IEEEdisplaynotcompsoctitleabstractindextext
\IEEEpeerreviewmaketitle

\section{Introduction}
As the fuel price is going up, alternative fuel vehicles, in particular electric vehicles, are getting more attentions. Previously, electric vehicles were suffering from various issues such as low reliability, high consumer satisfaction, and low return on investment~\cite{Romm2006}. However, thanks to the recent advancements in automotive, electronics, and communication technologies, electric vehicles have overcome the issues and have become pervasive on the present highways. Still, there are a number of challenges to deal with in order to make electric vehicles more practical. One of the most crucial problems is that the capacity of the state-of-art battery for electric vehicles is not sufficient to drive the cars over a long distance without recharging. It is well-known that the battery technology has been slowly improved \cite{Li2014}, and thus we can hardly expect to have a much higher capacity battery for electric vehicles in the near future. Meanwhile, as of today, there is a relatively small number of places to recharge the battery of an electric vehicle along a highway~\cite{Timpner2014}. Most of all, compared to the time to refuel a traditional vehicle with combustion engine, it takes a significant longer time to charge the battery through plug-in technology in a charging station~\cite{Li2012} or by staying on a wireless charging plate (CP)~\cite{Musavi2012}.

Recently, the concept of online electric vehicle\footnote{http://olev.kaist.ac.kr/en/index.php} is introduced to alleviate the aforementioned issues of electric vehicles, where electric vehicles with a special onboard unit can obtain the electricity on the move in a wireless manner while passing over a road surface under which power-line is installed. It is envisioned that this new technology will make a significant contribution to expand the adoption of electric vehicles. Despite the apparent benefits, there are still a number of issues to identify on the wide deployment of online electric vehicles. In particular, consider the problem of designing a proper billing scheme for this wireless-charging-on-the-move strategy. One straightforward solution would be taking a picture of every vehicle which is entering a road designed for online electric vehicles, and send the flat amount of bill to each of them. This strategy works well for many vehicles with combustion engine on modern toll highways where they pay their toll tax in a wireless and efficient manner. However, it is clearly not fair for online electric vehicles with almost fully charged battery to pay the same amount of money paid by those vehicles with almost empty battery since the driver of a vehicle with the fully charged battery may not want to use the online power service to save cost. In addition, collecting the pictures of the vehicles entering every part of the road is almost not possible and also can cause a privacy issue, and therefore, it is not desirable even though this strategy is widely used for toll tax collection at toll plaza on the highways. Moreover, recently radio frequency identification (RFID) gained a lot of attention from the service providers and has been widely deployed due to its simple operation and low cost~\cite{HZhu2011, Yalcin2010, WXie2013, Arco2011}. However, our case is completely different from RFID scenario because in our case the requirements for traceability and deniability are critical when compared to RFID authentication and billing. Hence the comparison between RFID-based solutions and our scenario would not be fair.

Motivated by our observations, in this paper, we propose a new secure-and-fair billing framework for online electric vehicles. In detail, we propose to adopt a road which consists of a series of short-and-equal-sized electricity supply segments, each of which serves as a unit for billing. Before an online electric vehicle enters a new segment, it can decide to use the electricity while moving over the segment, or deny it depending on its current battery level. Once decided to use, it needs to authenticate itself to the segment and obtain a secret to consume electricity (see Fig.~\ref{fig:roadinfra}). To improve the degree of fairness and service granularity, i.e. to be billed for actual use only, it is important to make the segment short. At the same time, to improve the efficiency of the segment, which is the actual rate of the segment used for charging against the portion of the segment used for other use such as authentication, it is crucial to design the authentication protocol to be as lightweight as possible. To guarantee the privacy of each driver and reduce the operation cost, it is highly desirable not to use camera for billing. Rather than using camera to take the picture of each vehicle to enter every segment, we propose a wireless communication based conditional privacy mechanism so that the real identity of the driver can be exposed only if there is a legal need such as refusing to pay after actual electricity consumption. Finally and most importantly, we provide a mutual audit mechanism through which a driver cannot deny the usage of legitimate electricity as well as the electricity service provider cannot overcharge.

\begin{figure}[t!]
\begin{center}
        \includegraphics[totalheight=3cm]{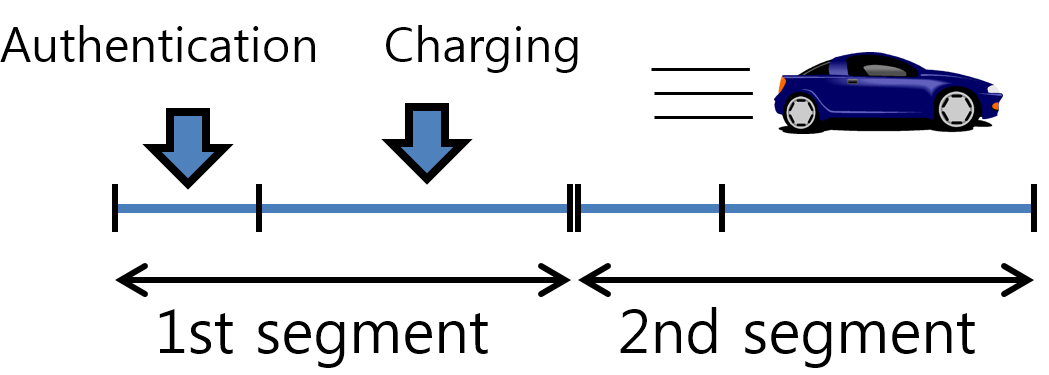}
    \caption{Online electric vehicle can charge its battery while moving after the authentication with the highway segment.}
    \label{fig:roadinfra}
		\end{center}
\end{figure}

In this paper, we propose a security framework for electric vehicles which support the following relevant features, and thus the contributions of this paper are:
\begin{enumerate}[a.]
\item \textit{Bidirectional audit}: we design a mechanism that guarantees privacy-aware bidirectional auditability for both electric power providing authority and the vehicles. To deal with the billing and auditing, we propose a semi-simultaneous billing where each vehicle, when it charges its battery, is billed on plate-by-plate basis where each plate delivers a constant amount of energy. In other words, the vehicles are billed with a fixed amount.
\item \textit{Conditional privacy}: we use multiple pseudonymous mechanism to preserve the conditional privacy of the vehicles at every stage of the protocol.
\item \textit{Mutual authentication}: We devise a fast and lightweight authentication mechanism for vehicles and the charging plates keeping in mind the portion of the charging plate designated for authentication.
\item \textit{Game-theoretic approach}: We employ a game theoretic approach to model the proposed bidirectional auditability mechanism in order to establish Nash Equilibrium between the charging plate and the vehicle.
\end{enumerate}

This paper proceeds as follows. Section II outlines the state of the art regarding wireless power transfer followed by the system model and problem definition in Section III. In Section IV, we outline our proposed scheme and analyze our system in Section V. In Section VI, we give our concluding remarks.

\section{State of the Art}
Today, some of world most renowned automobile companies are producing battery propelled vehicles and it is envisioned that soon electric vehicles will outclass the conventional automobiles due to economic and environmental reasons. The technological breakthrough in both electrification technology and the energy storage technology has made it possible for the automobile companies to achieve this milestone. From the studied conducted so far, it can be inferred that in the near future, most of the fossil fuel propelled vehicles will possibly be replaced by the electric vehicles (see \cite{Weissenger2010}). In \cite{Weissenger2010}, Weissenger et al. outlined the speed range, storage range, and the battery types of vehicles till the year 2008. To date, many efficient charging schemes have been proposed in the literature to save the commute time of the drivers \cite{Hoke2011}. However, the frequency of recharging is still a problem that needs to be addressed.

To motivate the use of electric vehicles, a new concept of wireless power transfer (WPT) was introduced \cite{Musavi2012}. In \cite{Musavi2012}, authors carried out a detailed survey regarding wireless power transfer and covered many dimensions such as the distance between the transmitting and receiving entity, cost of the technology and so forth. The concept of green car was introduced in 2009 by KAIST, South Korea by the name of online electric vehicle (OLEV) \cite{Jang2012}. The motivation for OLEV was the weight and the cost of the battery in electric vehicles, low frequency of charging, fast installation, low maintenance cost and so forth. To date, remarkable results have been achieved by this project and currently they run prototype buses in the KAIST campus South Korea \cite{Ko2013,Suh2013}. Nonetheless, such online vehicle would require massive power-line infrastructure installed under the road. Moreover coverage would be another issue due to the cost factor. From mutual authentication standpoint, Chuang et al. \cite{Chuang2011} proposed a hash-based authentication mechanism called trust-extended authentication mechanism (TEAM). TEAM adopts the concept of transitive trust relationships where a normal vehicle becomes the trusted entity after successful authentication and can delegate the authentication process in the absence of the authorities. Moreover TEAM does not protect the privacy since original ID is shared during authentication. On the other hand, even if a normal vehicle successfully authenticates itself, does not guarantee that it will not be malicious while delegating authentication function. Therefore, we believe that the transitive trust may lead to even worst situation from security standpoint in VANET.

Billing is an important requirement in commercial networks and it can be abstractly divided into two classes, time-based billing and content-based billing. In the former, nodes (subscribers or consumers) pay the service fee based on time, for instance the internet access charges and in the latter case, nodes pay based on the content they receive where the specific content costs a constant amount of money, for example downloading a song from iTunes and so forth~\cite{ YaoTMC2014}. A number of billing mechanisms have been proposed for wireless mesh networks~\cite{ HZhu2008, Lee2010} commercial VANET application~\cite{YaoTMC2014, YaoTIS2014}. In~\cite{YaoTMC2014}, the authors propose a portable authentication/authorization/accounting (AAA) framework for purchasing services from the RSUs. They use signature-based and key policy attribute-based encryption (KP-ABE) in their billing mechanism to attain localized fine-grained access control and also employ E-coin. In another work, Yeh et al.~\cite{ YaoTIS2014} propose a local and proxy-based authentication and billing scheme to lessen the long-distance communication overhead. They also propose an incentive-aware multi-hop forwarding for vehicles in the VANET. They use batch verification mechanism in their scheme to fulfill the security requirements and signature-based communications. However, our service scenario is different because we deal with the charging plates installed underneath the road and such sophisticated cryptographic primitives will cause enormous delay. Therefore aforementioned schemes are not directly applicable in our scenario.

In this paper, we to the best of our knowledge, for the first time propose a secure and privacy-aware mechanism to transfer the electric power to propel the vehicles moving on the road where the power transfer technology is installed underneath the road in the form of charging plates. Moreover our proposed scheme also guarantees bidirectional audit. First we propose two lightweight and fast privacy-aware mutual authentication mechanisms between the vehicles and the charging plates installed under the road. The two authentication mechanisms can be adapted with different vehicular speeds and the length of the charging plates. Then we propose a secure charging mechanism for vehicles with bidirectional auditability guarantee where vehicle is billed in a semi-simultaneous manner on the per-plate basis. We also employ a game theoretic approach for modeling and guaranteeing auditability by establishing Nash Equilibrium between the charging plates and the vehicles.

\begin{figure}[t!]
\begin{center}
        \includegraphics[totalheight=6cm]{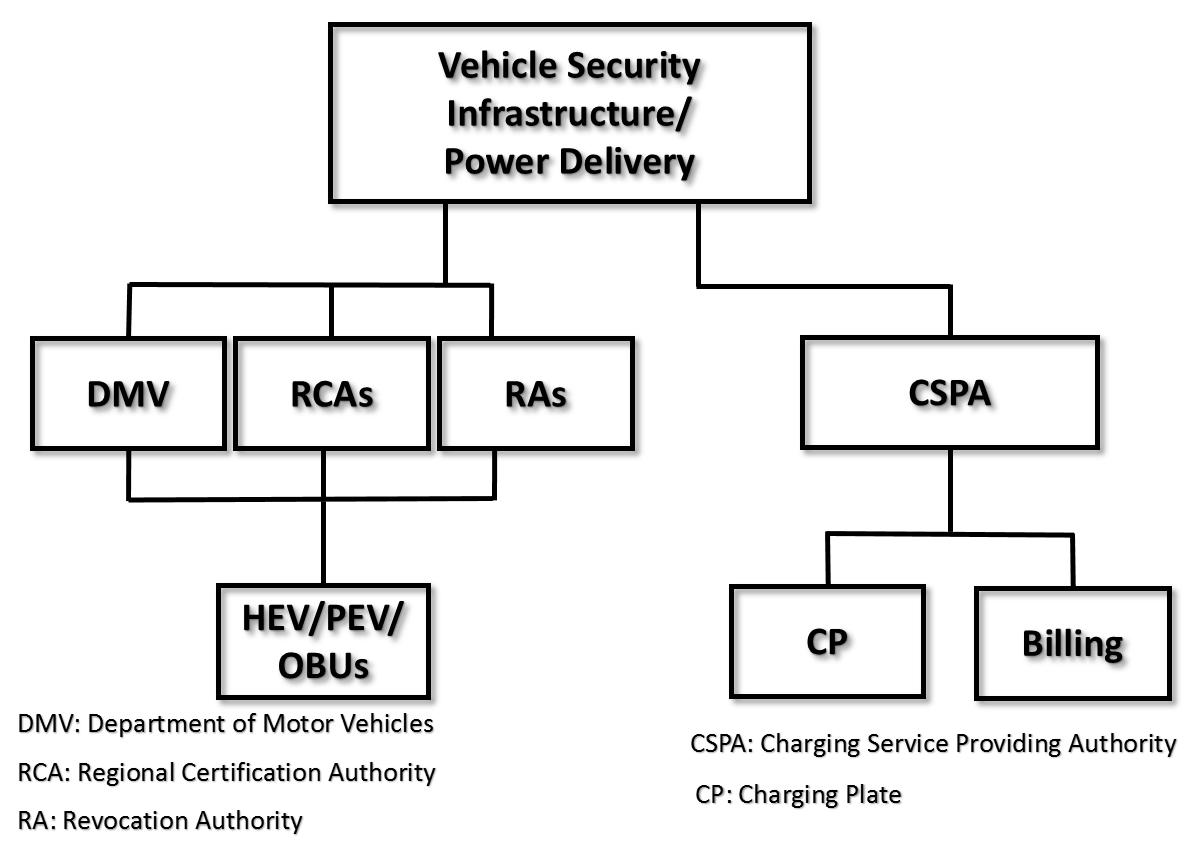}
    \caption{Taxonomy of System Participants}
    \label{fig:taxonomy}
		\end{center}
\end{figure}

\section{System Model and Problem Statement}

\subsection{System Participants and Network Model}
Our proposed system model consists of electric vehicles and an electric power delivery infrastructure. Electric power delivery service is exercised by the charging service providing authority (CSPA) that is responsible for providing the vehicles with electric charge through charging plates and bill them accordingly. The charging infrastructure is installed under the surface of the road and each road segment of certain length is covered by the charging plates. The charging plates also have a hardware for communication and computation purpose and these plates are responsible for authentication prior to battery charging, billing and logging the audit information. These charging plates communicate with both vehicles and CSPA back and forth during the charging and the billing process. We also introduce some components from vehicular ad hoc network (VANET) which are frequently assumed in VANET. These components are used by electric vehicles\footnote{Throughout the paper, the terms 'vehicle', 'vehicular node', and 'OBU' are used interchangeably and we mean electric vehicle collectively by these terms.} for initialization and registration. These components include vehicle management, registration and revocation authorities. The department of motor vehicles (DMV) is at the top of the hierarchy where every vehicle should be registered beforehand. Revocation authorities are leveraged to revoke the identity of the vehicle when needed with the consent from law enforcement authorities (police or judiciary) in the form of a warrant. There may be heterogeneous types of vehicles on the road, but for ease of understanding, we focus only on the electric vehicles in this paper. Therefore, our proposed scheme can be easily implemented in the VANET framework, which is one of the most popular and promising future vehicular infrastructure.

The taxonomy of the system participants is shown in Fig.~\ref{fig:taxonomy} and the network model is shown in Fig.~\ref{fig:networkmodel}. We consider a fleet of electric vehicles on the road where these vehicles receive electric power from the power line installed beneath the road, depending upon the usage of the vehicle. It can be seen in Fig 1 that road segment consists of the power line distribution technology installed beneath the road in the form of charging plates. A portion of the charging plate is leveraged for authentication purpose and the rest of the charging plate is used for transferring the electrical energy to the vehicle (see Fig.~\ref{fig:roadinfra}). Vehicles mutually authenticate each other with the charging plate before receiving electric power from the plates. After successful authentication, the fixed designated amount of electrical energy is transferred to the battery and the vehicle is billed accordingly. The communication channel between vehicles, registration and revocation authorities, and the charging plates is based on Dedicated Short Range Communication (DSRC)\footnote{[Online] http://www.iteris.com/itsarch/html/standard/dsrc5ghz.htm} standard whereas charging plates are connected to the charging service providers through high speed wired links. Our proposed scheme is based on the following assumptions.
\begin{enumerate}[a.]
    \item Electric vehicles are equipped with On Board Unit (OBU) and Tamper-Resistant Module (TRM) to carry out the secure computation.
    \item DMV is a trustworthy entity and only DMV is authorized to initialize the TRM and store necessary security parameters and keys in it, whereas CSPA, charging plates and OBUs are non-trustworthy.
    \item Every vehicle is also pre-loaded with a pool of pseudonyms (traceable by revocation authorities) for privacy reasons.
    \item Charing plates are installed under the designated road segments by CSPA and are equipped with hardware that is capable of carrying out secure computation and communication operations.
    \item Vehicles change their pseudonyms at every charging phase.
    \item The charging process is not automatic and it can be started with the consensus of the driver if the battery needs to be recharged.
    \item For a single charging plate, a fixed amount of charge is transferred to the battery and a fixed amount of bill is charged to the customer, and thus a semi-static and fixed billing system.
\end{enumerate}
Moreover the proposed billing framework must fulfill the following requirements.
\begin{enumerate}[\textit{R-}1]
\item While transferring electric power to the vehicle, bidirectional auditability must be guaranteed. In other words, the vehicle must not be over-billed and the CSPA must not be under-paid.
\item The conditional privacy of the vehicle's location and the user must be preserved. The identity of a vehicle owner should be revoked to the power supplier only if it is legally necessary, e.g. refuse to pay.
\item Due to the resource constraints of the charging plate and the speed of the vehicle, the communication between charging plate and OBU, and between charging plate and CSPA must be minimal. Moreover the authentication mechanism must be very fast and lightweight.
\item The billing procedure must be verifiable by all the entities, i.e. OBU, CSPA, and DMV.
\item At the time of charging, both the players should be in the Nash Equilibrium state.
\end{enumerate}

\begin{figure}[t!]
\begin{center}
        \includegraphics[totalheight=5.6cm]{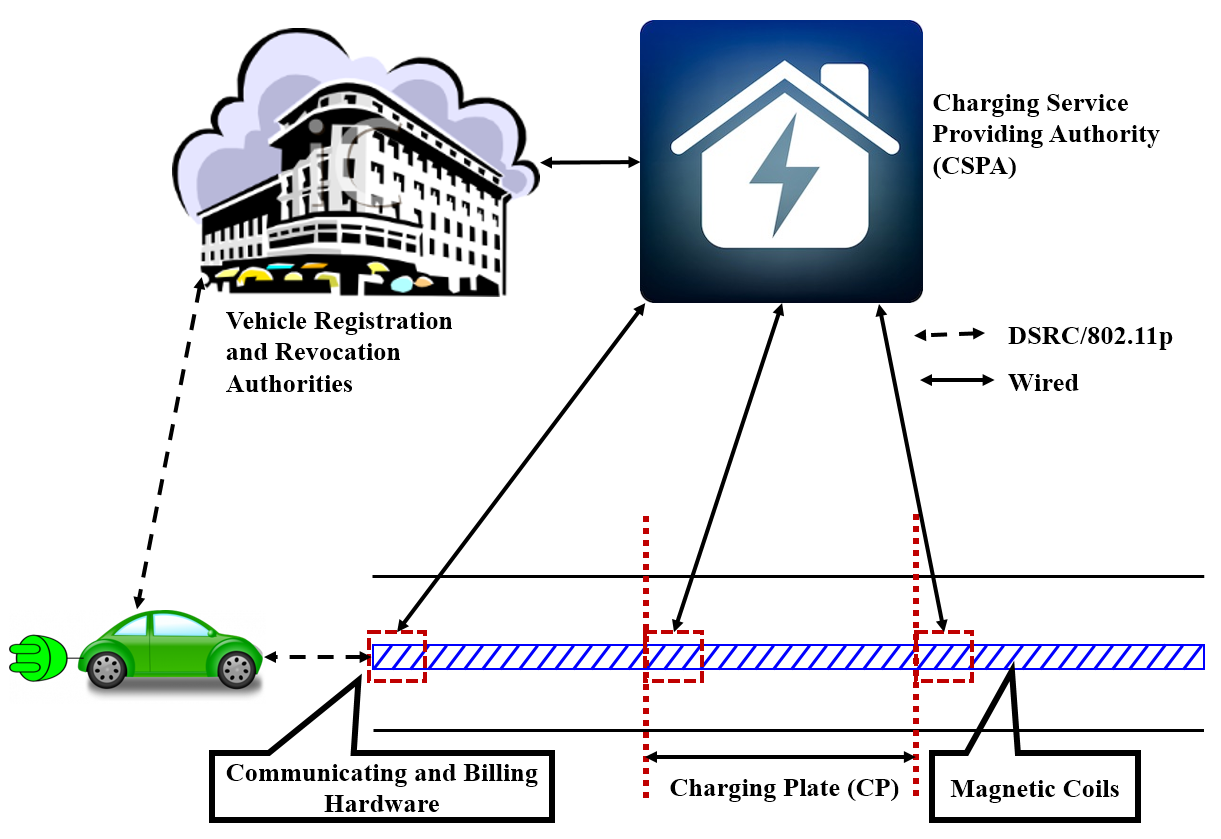}
    \caption{The Network Model}
    \label{fig:networkmodel}
		\end{center}
\end{figure}

\subsection{Threat Model}
In our threat model, we consider that both participating parties (charging plate and OBU) may be malicious. Their behavior can be malicious in terms of either bypassing the billing process or overcharging the energy receiving entity. Besides, the CSPA can also abuse the privacy of the electric power consumer vehicle by either exposing their location information or selling out their location-based profile to other third parties such as ads agencies and so forth. Moreover, the adversaries can sniff the communication between charging plate and OBU, modify it or forge it. We argue that the adversaries will have more resources than the participating entities. However, the timeliness of possible attack is a challenging front for the adversaries where the possible attacks must be performed within the stipulated time that is equal to the duration of the charging and billing.

\section{Proposed Bidirectional Auditability in Online Electric Vehicle}
In this section we outline the proposed power transfer and billing mechanism and the bidirectional audit between the charging plate and the OBU by a game-theoretic approach.

\subsection{Baseline}
  Before using the wireless electric energy transfer service on the move, vehicles must have registered with the DMV to initialize their TRM and to store the security parameters and pseudonyms in it. Additionally, the vehicles must also register with the CSPA to get the necessary security parameters, required to the charging plates at authentication stage. Whenever a vehicle\footnote{The term \lq \textit{vehicle}\rq throughout the rest of the paper should be read as electric vehicle. For the sake of simplicity we use the term vehicle instead of charging vehicles.} enters the road section with power line underneath it, it opts for either obtaining the electric power or not. If the vehicle selects the power reception, then it has to mutually authenticate with the charging plate. We propose two very fast and lightweight mutual authentication mechanisms, one is based on only hash and XOR functions and inspired by the Chuang et al.'s scheme \cite{Chuang2011}, while the second one is based on the hash chain. The former is a direct authentication between charging plate and OBU whereas the latter is authentication through CSPA. In the former scheme, charging plate incurs minimum communication delay whereas in the latter, charging plate incurs minimum computation delay. Both of the proposed scheme are suited for specific purposes that are explained in the paper. After successful authentication, the power transfer process starts and charging plate sends the billing information to both OBU and CSPA. The billing is fixed on per charging plate basis. We also model the billing and audit as a game between OBU and charging plate where both of them must achieve the Nash Equilibrium state.

\subsection{Preliminaries and Initializations}
\subsubsection{System Initialization}
Notations in Table~\ref{table:notations} are employed throughout the paper. We use pseudonymous approach for privacy preservation. In addition, in order to store the individual secret keys of the vehicles, i.e. $K_{sym}$ and $K_{V}$ in the revocation authorities (RAs), we use ElGamal encryption algorithm over elliptic curve cryptography (ECC) due to its proven security.  Let $\mathbb{G}$ be a cyclic group of prime order $q$ where $\mathbb{G}$ is generated by a generator $P$. First of all DMV chooses a random number $x\in\mathbb{Z}^{*}$ as its private and computes $PK^{+}=xP$ as its public key. DMV then uses threshold based secret share scheme \cite{Zhang2008} and divides $x$ into $j$ parts where $j$ is the number of revocation authorities, each $RA_{i}$ holds a share $x_{i}$ and $x_{i}\in(x_{1},x_{2},x_{3},...,x_{j})$. In order to construct $x$ from individual $x_{i}$, RAs must elect one of them to be group leader and construct $x$ from combination of $x_{i}$. For the selection of group leader, any available efficient group leader election mechanism in the networks can be used.

\begin{table*}
\label{tabel1}
\renewcommand{\arraystretch}{1.1}
    \caption{Legend for symbolic notations}\label{table:notations}
		\centering
    \begin{tabular}{|c|c|}
		\hline
      \textbf{Notation} & \textbf{Explanation}\\
			\hline
			\hline
					$\mathbb{G}$ & Cyclic group of Order $q$\\
			\textit{P} & The generator of $\mathbb{G}$\\
			\textit{r} & Random nonce\\
			\textit{x,$x_{i}$} & Private master key and \textit{i}-th share of \textit{x}\\
			$PK^{+}$ & Public key corresponding to \textit{x}\\
			\textit{$K{_{DMV}^+}$,$K{_{DMV}^-}$} & Public private key pair of DMV for signing pseudonyms\\
			\textit{$c_{V}$} & Vehicle V's secret initial counter used in pseudonym generation\\
			\textit{$inc_{V}$} & Incrementing factor for pseudonyms \\
			\textit{$K_{sym}$} & Vehicle $V$'s  AES symmetric key used in pseudonym generation\\
			\textit{$K_{V}$} & $V$'s individual secret key \\
			\textit{$PS{_{OBU}^i}$} & Vehicle \textit{V}'s \textit{i}th pseudonym \\
			\textit{MSK} & Hash Chain based Master secret key \\
            \textit{$X_{OBU}$} & Hash of the overall pseudonym pool \\
            \textit{$Cert_{OBU}$} & OBU's anonymous certificate issued by a certification authority \\
            \textit{$PWD_{OBU}$} & OBU's initial password to log in to the system in order to start registration\\
			\textit{$H(\cdot)$} & A MaptoPoint hash function as $H:\{0,1\}^* \rightarrow \mathbb{G}$ \\
            \textit{$h(\cdot)$} & Collision-resistant hash function \\
			\textit{$h_{k}(\cdot)$} & Keyed hash function \\
			\textit{$\oplus$} & Exclusive OR operation \\
			$\mid\mid$ & Concatenation function\\
			\hline
    \end{tabular}
\end{table*}

\subsubsection{TRM Installation}
Only DMV is authorized to install the TRM in the vehicle for the first time after purchase or re-purchase. The owner of the vehicle has to personally visit the DMV for the installation and/or initialization of the TRM. After confirming the credentials of the vehicle and its owner, DMV initializes TRM and saves the system parameters in the TRM including $(\mathbb{G},q,P,PK^{+},c_{V},inc_{V})$. Additionally DMV also preloads TRM with vehicle's individual secret key $K_{V}$ and pseudonym generation key $K_{sym}$.

\subsubsection{Pseudonyms Assignment}
DMV generates $n$ number of pseudonyms for each vehicle by taking vehicle’s secret counter $c_{V}$ and increment it by vehicle $V$'s incrementing factor $inc_{V}$. The pseudonyms are generated as follows:
$PS{_{OBU}^i}=\{{(\alpha)_{K_{sym}} \|(\alpha\oplus ID)_{K_{V}} \|n_{i} }\}_{K{_{DMV}^-}}$  where $\alpha=c_{V}+n_{i}\cdot inc_{V}$, $n_{i}$ is the current count of generated pseudonym (note that it may not be linear), and $ID$ is the vehicle's identity. Then DMV stores these pseudonyms in its database and indexes it with the value of $n$. After all pseudonyms are generated for the vehicles, DMV saves these pseudonyms in vehicle's TRM along with another value $X_{OBU}=h(PS{_{OBU}^1},PS{_{OBU}^2},...,PS{_{OBU}^n})$ and sends the anonymous pseudonyms to RAs as well. In order to help in revocation, TRM also encrypts both $K_{sym}$ and $K_{V}$ and sends it to RAs which serve as a trapdoor in revocation. The aforementioned keys are encrypted with public master key using ElGamal encryption as follows:
$$
    \begin{array}{l}
     \delta_{1}=rP, \delta_{2}=(K_{sym}\|K_{V})\oplus H(rPK^+)
    \end{array}
$$
 $r$ is a random nonce selected by the TRH for this encryption, then it sends $\{\delta_{1},\delta_{2}\}$ to RAs. However RAs can only decrypt the keys $K_{sym}$ and $K_{V}$ when they have a warrant to do so after colluding to construct $x$ from individual $x_{i}$. The reason for saving encrypted keys in RAs’ database is twofold: RAs use these keys to revoke a vehicle in case of any dispute and for privacy reasons; we do not want RAs to link pseudonyms and/or extract $c_{V}$ and $inc_{V}$ from the beacons until necessary, otherwise.

 It is also to be noted that when a vehicle consumes all the pseudonyms it has in the pseudonym pool, it needs to obtain a batch of fresh pseudonyms from the DMV. The vehicle does not need to be physically present at DMV, rather it can obtain the pseudonyms from DVM by connecting through internet. We assume that the existing pseudonym refilling strategies can be used~\cite{ Ma2008, Benin2010, Petit2012, Mahmoud2014}. Now we outline the two mechanisms for mutual authentication and billing.

\subsection{Direct Mutual Authentication (DMA)}

In the direct approach, OBU and CP mutually authenticate each other without intervention of the CSPA. First of all CSPA creates $l$ number of master secret keys MSK based on hash chain by selecting a secret $s$ where $MSK_{i}=h^{i} (s)$ and sends the key to DMV as follows:
$$
    \begin{array}{l}
     CSPA\rightarrow DMV: MSK_{i}(i=1,2,3,...,l)
    \end{array}
$$
$MSK_{i}$ is a hash chain based master secret key which is based on a secret $s$ and $MSK_{l}=h^{l} (s)$. In other words, each $MSK_{i}$ is used for a designated amount of time, and CSPA updates $MSK_{i}$ after regular intervals. Since $MSK_{i}$ is distributed by CSPA, it can be updated in timely manner by CSPA and vehicles will receive the updated $MSK_{i}$ in their next registration phase with CSPA. After that, DMV also sends $X_{OBU}$ of the registered vehicles to CSPA for records.
$$
    \begin{array}{l}
     DMV\rightarrow CSPA: X_{OBU}
    \end{array}
$$
Each vehicle has a pool of legitimate traceable pseudonyms from DMV and at the time of authentication, it can use any pseudonym to start charging. The vehicle will be billed based on the used pseudonym.

\subsubsection{Vehicle Registration with CSPA}
The vehicle, most precisely its OBU, must register with CSPA before charging. We assume a secure channel between CSPA and the vehicle. The registration of the vehicle proceeds as follows. The vehicle starts with the password and upon access, the CSPA calculates some security parameters for the vehicle and sends it back to the OBU. The different steps and their descriptions are given below:

\begin{enumerate}[a.]
\item $OBU\rightarrow CSPA: PWD_{OBU},X_{OBU}$. OBU sends these values to CSPA on a secure channel. If $PWD_{OBU}$ is valid, then the protocol will proceed.
\item CSPA calculates the following 3 values, i.e. $H_{1},H_{2}$, and $H_{3}$. $H_{1}$ is used as a secure parameter kept by CSPA. $H_{2}$ and $H_{3}$ are the authentication parameters that are sent back to the OBU.
   $$
    \begin{array}{l}
     H_{1}=h(s\|X_{OBU})\\
     H_{2}=h^{2}(s\|X_{OBU})\\
     H_{3}=MSK\oplus H_{1}
    \end{array}
$$
\item CSPA registers the OBU by sending the hash function $h( )$, $H_{2}$, and $H_{3}$ to the OBU and recording these parameters by storing it in its database against the value of $X_{OBU}$.\\
    $CSPA \rightarrow OBU: X_{OBU}, h( ), H_{2}, H_{3}$.
\end{enumerate}

\begin{figure*}[t!]
\begin{center}
        \includegraphics[totalheight=7cm]{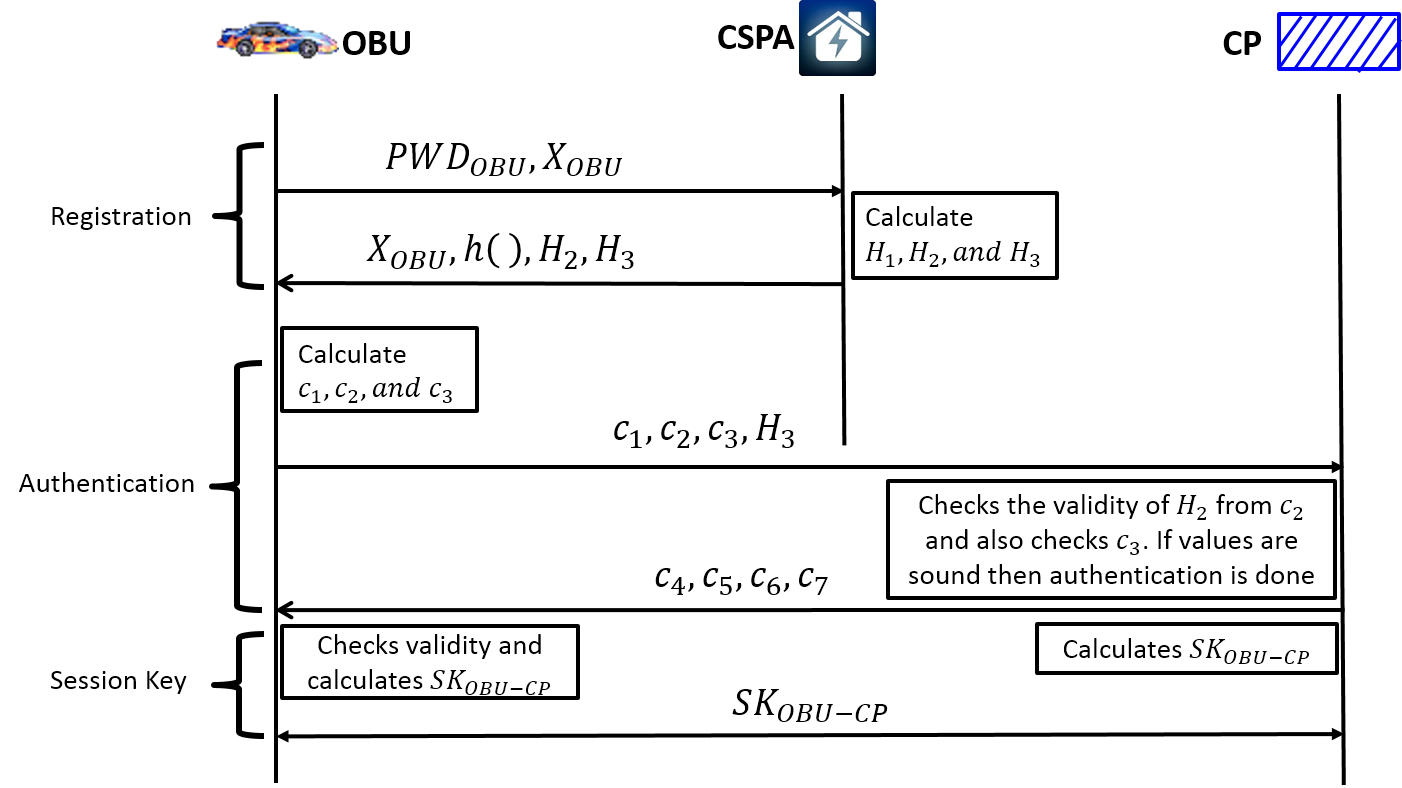}
    \caption{Authentication process in DMA scenario}
    \label{fig:dma}
		\end{center}
\end{figure*}

\subsubsection{Authentication between OBU and CP}
After completing the registration phase with CSPA, when the vehicle passes through the section of the power-line-enabled road, it starts authentication process with the charging plate. The vehicle starts by selecting a pseudonym from its pool of pseudonym and proceeds the whole electric power receiving process with the selected pseudonym anonymously. The partial per charging bill is prepared based on the presented pseudonym in the authentication/charging process. It is to be noted that a fixed amount of electric power is delivered to the vehicle's power reception module that costs a fixed amount. The comprehensive mutual authentication steps are given below:
\begin{enumerate}[a.]
\item OBU selects a pseudonym $PS{_{OBU}^i}, i=1,2,3,...,n$ from its pool and calculates the following parameters.
    $$
    \begin{array}{l}
     c_{1}=h(H_{2})\oplus PS{_{OBU}^i}\\
     c_{2}=h(PS{_{OBU}^i})\oplus X_{OBU}\\
     c_{3}=h(h(PS{_{OBU}^i})\|c_{2}\|H_{3})
    \end{array}
$$
\item Then OBU sends CP, the values calculated in previous step along with $H_{3}$.
    $$
    \begin{array}{l}
     OBU \rightarrow CP: c_{1}, c_{2}, c_{3}, H_{3}
    \end{array}
$$
\item CP executes the following steps.
\begin{enumerate}[i.]
\item Start with $H_{3}$ and extract the secret $H_{1}$ as $MSK\oplus H_{1}\oplus MSK$.
\item It calculates $H_{2}$ and extracts $PS{_{OBU}^i}$ from $c_{1}$.
    \item Then it checks for the value $c_{2}$ if it is equal to $h(PS{_{OBU}^i})\oplus X_{OBU}$.
        \item And check if $c_{3}$ is equal to the retrieved values $h(h(PS{_{OBU}^i} )\|c_{2} \|H_{3})$ then the OBU is authenticated, otherwise the authentication fails. It is to be noted that there will be a fixed number of tries, failing which will halt the authentication process.
\end{enumerate}
\end{enumerate}

After successful authentication, OBU and CP initiate the protocol to construct a session key $SK_{OBU-CP}$  which is used for the later communication and billing parameters. The initialization of session key from CP serves as an acknowledgement to authentication as well. The OBU will not have been authenticated, otherwise. CP extracts $X_{OBU}$ from the received messages $c_{2}$ and selects its nonce as $r_{cp}$ and calculates session key as $SK_{OBU-CP}=h(PS{_{OBU}^i} \|r_{cp})$. CP also calculates the following parameters.
$$
    \begin{array}{l}
     c_{4}=ID_{cp}\oplus r_{c}\\
     c_{5}=r_{c}\oplus h(h(PS{_{OBU}^i}))\\
     c_{6}=h(r_{c}\|c_{4}\|c_{5})\\
     c_{7}=H_{1}\oplus h^2(PS{_{OBU}^i})
    \end{array}
$$

$ID_{cp}$ is the ID of the charging plate. After calculating the above values, CP constructs an authentication reply message and sends it back to OBU. This authentication reply means that OBU has been authenticated and other parameters will be sent for the session key calculation.  The following authentication reply message is sent to OBU.
$$
    \begin{array}{l}
     CP\rightarrow OBU: c_{4}, c_{5}, c_{6}, c_{7}
    \end{array}
$$

From the above reply message, OBU extracts $r_{c}$ from $c_{4}$ and checks if $c_{6}$ is equal to $h(r_{c} \|c_{4} \|c_{5})$. If the information is correct, then the OBU authenticates CP as well and computes the session key $SK_{OBU-CP}=h(PS{_{OBU}^i} \|r_{c})$, extracts $H_{1}$ from $c_{7}$ and stores it as a security parameter.

At this point in time, the mutual authentication is completed and the electric power reception process will start based on the established session key $SK_{OBU-CP}$.

When these two entities (OBU and CP) authenticate each other then the vehicle will receive the designated power from the road (charging plate). At the end of each phase of the charging plate at $CP_{i}$, a unit cost $C_{i}$ will be accumulated to the account of the OBU against its presented $PS{_{OBU}^i}$. At the end of the whole power transfer from a number of charing plates, both OBU and CSPA will have the log of the amount of transferred power and the OBU will be billed accordingly which will be verifiable by both CSPA and the OBU. The whole authentication process in case of DMA, is shown in Fig~\ref{fig:dma}.

\begin{figure*}[t!]
\begin{center}
        \includegraphics[totalheight=7cm]{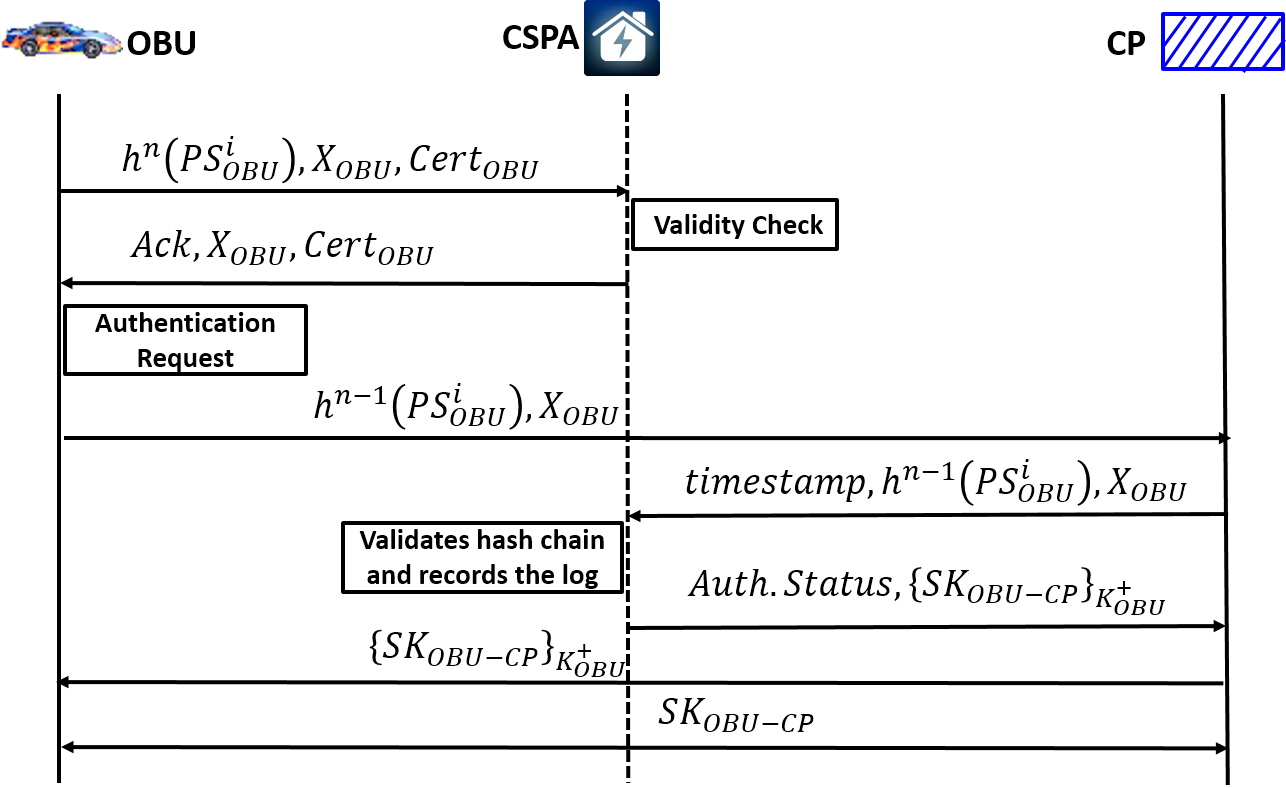}
    \caption{Authentication process in PHA scenario}
    \label{fig:pha}
		\end{center}
\end{figure*}

\subsection{Pure Hash Chain based Authentication (PHA)}
For a fast moving vehicle the DMA approach can be applied where computation is done locally by the charging plate; however, DMA may incur reasonable computation delay.  Therefore we propose another indirect authentication mechanism based on hash chain carried out by CSPA where computation delay is minimum whereas a small communication delay is introduced. Moreover this mechanism is most favorable for low speed vehicles. DMV provides the OBU with $n$ pseudonyms $PS{_{OBU}^i},i=1,2,3,...,n$, and hash chain corresponding to each pseudonym $h(PS{_{OBU}^i} ),h^{2}(PS{_{OBU}^i}),...,h^{n}(PS{_{OBU}^i})$. We assume that for the sake of charging the vehicle's battery, the vehicle registers with the CSPA based on the policy. In other words, the car uses a new hash chain based on the designated interval in the policy. Authentication process takes place as follows:
\begin{enumerate}[a.]
\item The vehicle registers with CSPA and sends one of the hash chain head to CSPA $h^{n} (PS{_{OBU}^i})$.
    $$
    \begin{array}{l}
     OBU\rightarrow CSPA: h^{n}(PS{_{OBU}^i}), X_{OBU}, Cert_{OBU}
    \end{array}
$$
\item 	At the time of authentication and request for charging, the vehicle must provide the CP with a member hash from the registered hash chain $h^{n-1} (PS{_{OBU}^i})$.
    $$
    \begin{array}{l}
     OBU\rightarrow CP: h^{n-1}(PS{_{OBU}^i}), X_{OBU}
    \end{array}
$$
\item CP forwards this value to the CSPA.
$$
    \begin{array}{l}
     CP\rightarrow CSPA: timestamp, h^{n-1}(PS{_{OBU}^i}), X_{OBU}
    \end{array}
$$
\item CSPA validates the hash, checks if $h(h^{n-1}(PS{_{OBU}^i}))=h^{n}(PS{_{OBU}^i})$ and replies accordingly. CSPA also replaces $h^{n}(PS{_{OBU}^i})$ with $h(h^{n-1} (PS{_{OBU}^i}))$. In addition to authentication, CSPA also provides the CP with a session key $SK_{OBU-CP}$ and saves it in its database with time and the $X_{OBU}$. It is to be noted that, CSPA issues a single session key for all the plates for a particular vehicle and a particular hash chain.
    $$
    \begin{array}{l}
     CSPA\rightarrow CP: Auth.Status, \{SK_{OBU-CP}\}_{K{_{OBU}^+}}\\
     CP\rightarrow OBU: \{SK_{OBU-CP}\}_{K{_{OBU}^+}}
    \end{array}
$$
    \item If the authentication is successful then charging plate will transfer the electric power to the vehicle, otherwise the process halts. It is to be noted that one hash chain is long enough to use it for the whole day. For the next day, the vehicles can register another hash chain. This process will still preserve the conditional privacy of the OBU.
\end{enumerate}
The protocol for pure hash chain based authentication mechanism is given in Fig~\ref{fig:pha}.

\subsection{Power Transfer, Billing, and Auditability}
Once the mutual authentication is completed between the charging plate and OBU, the vehicle starts to receive the power from the road (charging plate) semi-simultaneously with billing. More precisely, the vehicle is billed at a charging plate level. At the end of the charging process, the total bill is accumulated both at OBU and the CSPA. The process is explained in case of both direct and indirect authentication as shown in Fig~\ref{fig:ptindma} and Fig~\ref{fig:ptinpha}.

\subsubsection{Online Electric Power Transfer and Billing in DMA}
In case of DMA, CSPA has access to the pseudonym used in the authentication process through CP. Therefore after a successful authentication and establishment of a session key, vehicle requests for the electric power and presents the charging palte and CSPA with the pseudonym and other parameters for the billing purpose. The series of steps are given below:
$$
    \begin{array}{l}
     OBU\rightarrow CP:\\
     \{timestamp\|Charging Req.\|PS{_{OBU}^i}\|h_{K_{V}}(\alpha)\}_{SK_{OBU-CP}},\\
     \alpha=(timestamp\|Charging Req.\|PS{_{OBU}^i})\\
     CP\rightarrow OBU: \{Ack, timestamp\|PS{_{OBU}^i}\|h_{K_{V}}(\alpha)\}
    \end{array}
$$
At this point the charging plate starts billing and sends the bill log to OBU and to CSPA. The bill is logged against two values, $X_{OBU}$ and the consumed pseudonym $PS{_{OBU}^i}$.
$$
    \begin{array}{l}
     CP\rightarrow CSPA:(timestamp\|PS{_{OBU}^i}\|X_{OBU}\|Cost_{CP_{i}})\\
     CSPA: Cost_{X_{OBU}}=\sum_{i=1}^{n}Cost_{CP_{i}}
    \end{array}
$$
In case of DMA, the CSPA accumulates all partial billing information from individual charging plates and bill the OBU accordingly. It is worth noting that we use a constant cost per charging plate.

\subsubsection{Online Electric Power Transfer and Billing in PHA}
In case of PHA, the CSPA does not have access to individual pseudonyms; rather it maintains the billing information based on the $X_{OBU}$ value. After successful authentication, OBU requests for online power transfer and the electric power transfer begins. Meanwhile, CSPA bills the cost for the current CP and accumulates it to the account against $X_{OBU}$ with timestamp information. The series of steps are given below:

$$
    \begin{array}{l}
     OBU\rightarrow CP:\\
     \{timestamp\|Charging Req.\|X_{OBU}\|h_{K_{V}}(\beta)\}_{SK_{OBU-CP}},\\
     \beta=(timestamp\|Charging Req.\|X_{OBU})\\
    \end{array}
$$
The charging plate forwards the request to CSPA where CSPA constructs a reply message for the OBU. Meanwhile when charging plate receives the reply message from CSPA, the vehicle will start receiving power and CSPA will record the cost for the current charging plate.
$$
    \begin{array}{l}
     CSPA\rightarrow CP:(timestamp\|X_{OBU}\|h_{K_{V}}(\beta))
    \end{array}
$$
Charging plate forwards the message to OBU accordingly whereas CSPA calculates the bill for current charging plate and accumulates with the partial bills from previous charging plates.
$$
    \begin{array}{l}
     CSPA: Cost_{X_{OBU}}=\sum_{i=1}^{n-1}Cost_{CP_{i-1}}+Cost_{CP_{i}}
    \end{array}
$$
This way, CSPA calculates the bill partially simultaneously with charging process which is the motive of our GHP game. We will discuss our OBU-CP game in the next subsection.

\begin{figure*}[t!]
\begin{center}
        \includegraphics[totalheight=7cm]{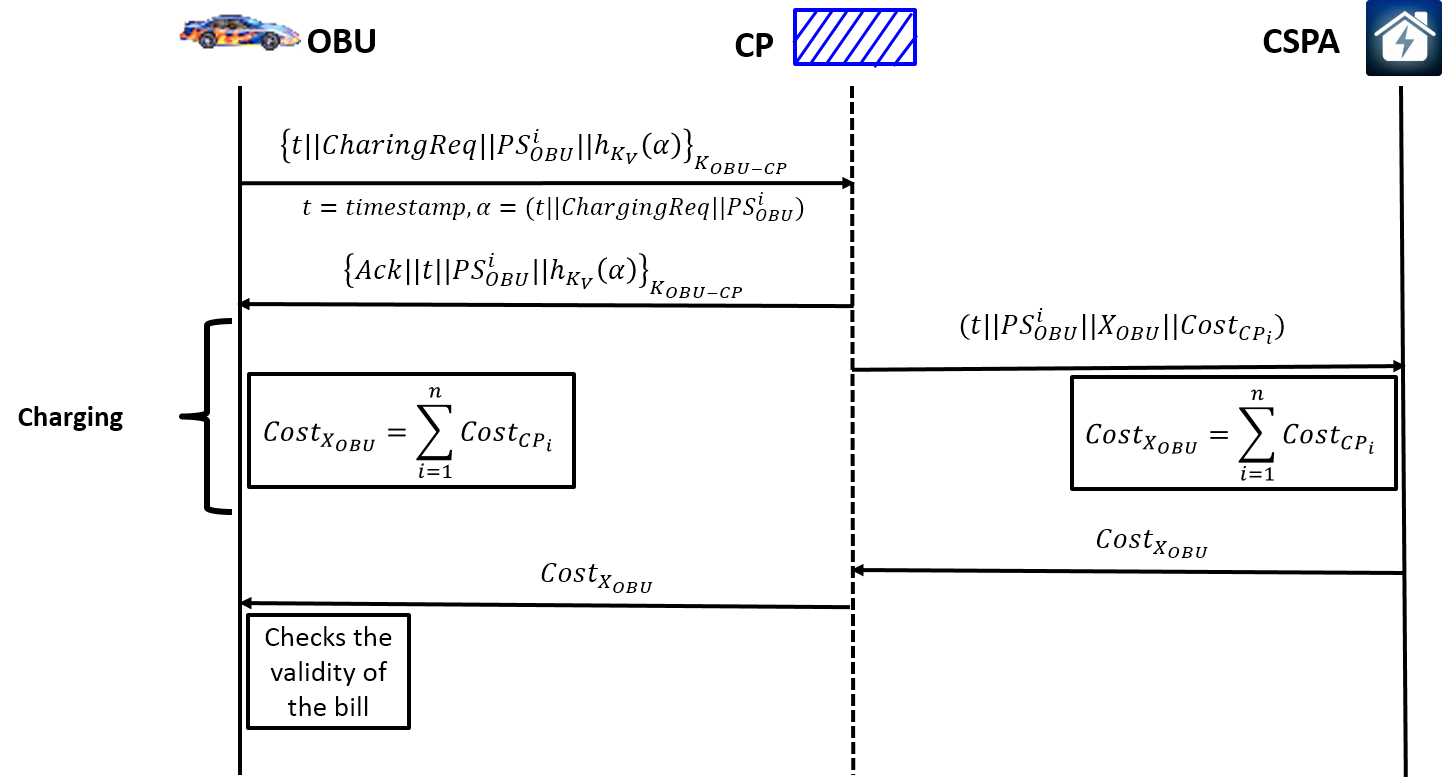}
    \caption{Power transfer and billing process in DMA scenario}
    \label{fig:ptindma}
		\end{center}
\end{figure*}

\subsection{Bidirectional Auditability Game}
We formalize the bidirectional auditability as an instance of the Guest-Host problem (GHP) with game-theoretic approach. In the GHP, the guest wants to use the hotel for a few days and does not want to get a bill (no) greater than actual use. Whereas the host wants to charge the bill to the guest for the actual (or more) use but wants to make sure that the guest does not deny the actual use. In our scenario, OBU can be assumed as guest and the charging plate as the host.

We explain the bidirectional auditability between OBU and charging plate with the help of an uncooperative game $\mathcal{G}$ which is defined as a triplet $(\mathcal{P}$,$\mathcal{S}$,$\mathcal{PO})$. $\mathcal{P}$ is the set of players of the game, $\mathcal{S}$ is the set of strategies followed by the players and $\mathcal{PO}$ is the set of payoff functions as a result of the players' strategies.
\subsubsection{Players}
The set of players $\mathcal{P}=\{OBU,CP\}$ corresponds to the set of OBUs and the charging plate (CP). There can be multiple OBUs serviced by the CP but we assume that at certain instant of time, only single OBU will be entertained at the start of the CP. Without loss of generality, multiple OBUs can recharge their batteries after successful authentication.

\subsubsection{Strategy/Move}
In our game, each of the two players follow two strategies, either $Cooperate (C)$ or $Deviate (D)$ and thus $S_{i}=\{C,D\}$. Moreover in cooperation state, each player makes a move that produces a better payoff. By cooperating, a vehicle changes its pseudonym every time is charges the battery and the CP bills it accordingly. Whereas in case of deviation, OBU misbehaves and does not follow the protocol or CP overcharges the bill against OBU.

\subsubsection{Payoff Function}
We formulate a payoff function for both players of the game. The payoff function $\mathcal{PO}(t)$ is given by:
$$
    \begin{array}{l}
    \mathcal{PO}_{i}(t)=a_{i}(t)-cost_{i}(t)
    \end{array}
$$
$a_{i}(t)$ is the advantage of player $i$  at time $t$ and $cost_{i}(t)$ is the cost of achieving $a_{i}(t)$. It is to be noted that $a_{i}(t)$ depends upon the successful battery charging and the normal billing and $cost_{i}(t)$ depends upon the pseudonym change for charging and the authentication overhead for both CP and OBU.

In game $\mathcal{G}$, the players do not know the strategic behavior of the opponent unless the complete billing has been done. Since we have only two strategic behaviors from the set $S_{i}=\{C,D\}$, there is $50\%$ probability for the players to guess the behavior of the opponent keeping in mind its payoff.

\textit{\textbf{Definition}: The best response $br_{i}(S_{j\in [C,D]})$ on the part of a player $i$ is a move such that:}
$$
    \begin{array}{l}
    br_{i}(S_{j\in [C,D]})=\max(\mathcal{PO}_{i}(s_{i}))
    \end{array}
$$

In other words, the best response of a player $i$ will be such that it results in a maximum payoff. If the two players unknowingly strategically give best responses to each other in the game, then the opponent will not have any chance to deviate from the game and the result of the game is called Nash Equilibrium (NE). When the game reaches an NE, then the players cannot increase their payoff by changing their strategy or deviating from the best response strategy.

\subsection{Nash Equilibrium in $\mathcal{G}$}
In NE, every player plays its best move and correctly anticipates that its opponent will do the same. In Table \ref{table:strategies}, we outline the possible moves made by each players.

\begin{table}
\label{tabel2}
\renewcommand{\arraystretch}{1.1}
    \caption{Strategic moves of the players in the game $\mathcal{G}$}\label{table:strategies}
		\centering
    \begin{tabular}{|c|c|c|}
		\hline
      \textbf{OBU $\backslash$ CP} & \textbf{C} & \textbf{D}\\
			\hline
			\textbf{C}	 & $1,1$ & $1,0$\\
            \hline
				\textbf{D}	 & $0,1$ & $-1,-1$\\
			\hline
    \end{tabular}
\end{table}

In our game, there is only one NE which is achieved through $(C,C)$. It is worth noting that a game may have more than one Nash Equilibria depending upon the nature of the game. In our game, the best strategy for OBU is to choose \textquoteleft$C$\textquoteright ~no matter what charging plate chooses between \textquoteleft$C$\textquoteright  and \textquoteleft$D$\textquoteright. This is because the only way for OBU to maximize its payoff in the form of battery charge and fair billing is to choose \textquoteleft$C$\textquoteright  at the expense of the cost incurred by changing pseudonym and shared authentication overhead. OBU may not know the strategy of the charging plate. The \textquoteleft$D$\textquoteright ~strategy will cause the loss which is unfair auditability leading to revocation for both CSPA and the OBU. Therefore from the strategic Table 2, $(1,1)$ is the best response from both side, where they cannot increase their payoff by changing their strategy.

\begin{figure*}[t!]
\begin{center}
        \includegraphics[totalheight=7cm]{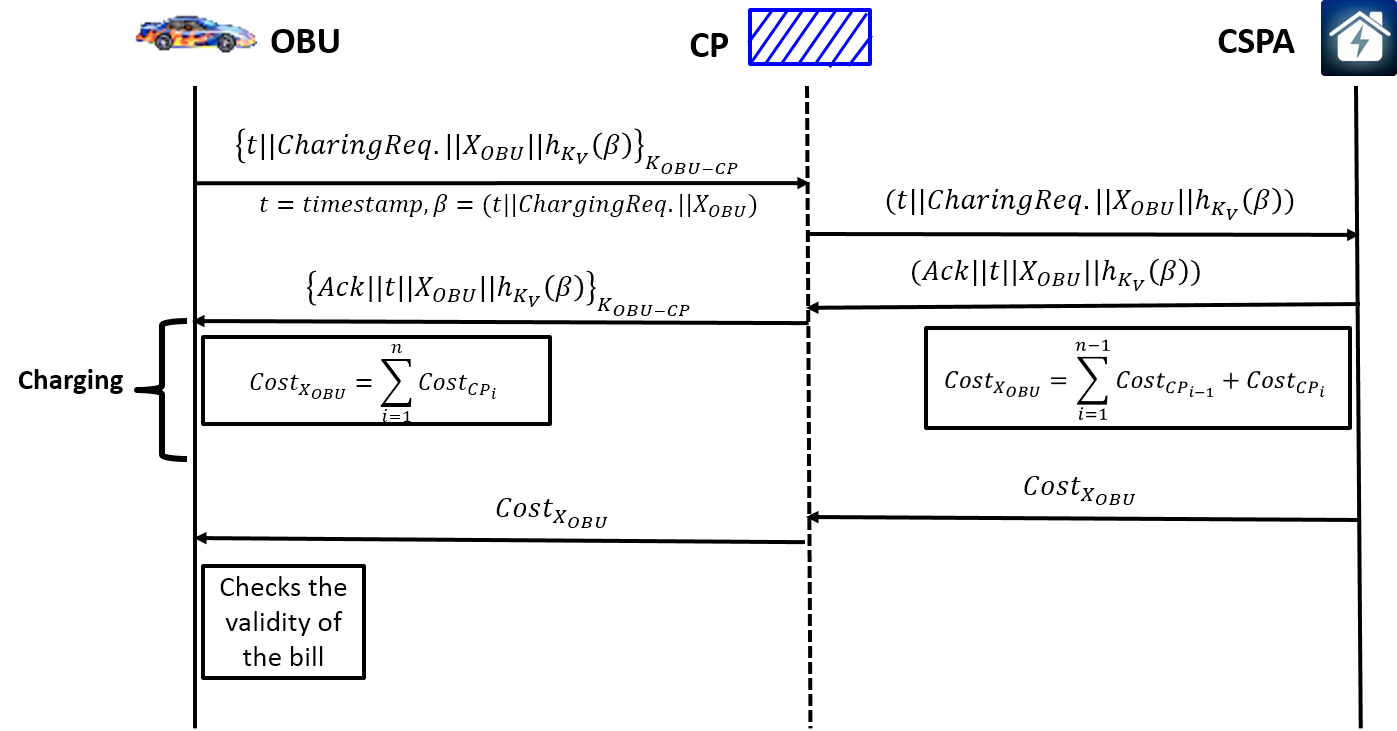}
    \caption{Power transfer and billing process in PHA scenario}
    \label{fig:ptinpha}
		\end{center}
\end{figure*}

\section{Evaluation and Analysis}

\subsection{Security and Privacy Analysis}
The security of our proposed scheme depends upon the collision resistance property of the one-way hash function. Given any $m$, it is easy to compute $h(m)$, and computationally very difficult to calculate the value of $m$ from $h(m)$. The most basic security requirement of our proposed scheme is mutual authentication between charging plate and OBU. With our proposed lightweight authentication protocol which is an extended version of Chuang et al.'s \cite{Chuang2011} protocol, mutual authentication is guaranteed before starting the charging process. If the underlying hash mechanism is secure, then our proposed authentication can be considered secure. However the effect of keys compromise can be critical for our proposed scheme. From the OBU perspective, compromising $K_{V}$ does not have dire consequences because the adversary $\mathcal{A}$ can get only a part of pseudonym, not the whole pseudonym. In case of the compromising both $K_{sym}$ and $K_{V}$, $\mathcal{A}$ can not only manipulate pseudonyms, but can reuse them. Moreover our system can prevent the secret sharing attack because a vehicle must authenticate itself prior to charging its battery. Since the charging plate authenticates the user, therefore two users cannot use the same secret and/or re-use it because timestamp and the local log of the usage of charging plate or CSPA will stop the vehicles from doing so. In other words, the protocol must follow the steps: (i) registration, (ii) authentication and (iii) power transfer.

Secure bidirectional auditability is provided through semi-simultaneous billing procedure incorporated in our proposed scheme. Each charging plate is capable to transfer a fixed amount of electrical energy to the battery and bill the vehicle with a fixed amount. Since vehicles use either individual pseudonyms $PS{_{OBU}^i}$, or the combined hash value $X_{OBU}$, the final bill is the combination of the costs of individual charging plates. Each OBU knows the cost per charging plate, and it records the cost in its log as well. Both CSPA and OBU can verify the individual and final bill of the power transfer in a liable and a non-repudiate manner. It is also worth noting that it is the duty of CSPA to make sure the freshness of the session key and use different session keys in different charging plates for security reasons. The behavior of both charging plate and OBU from security perspective is depicted by the game $\mathcal{G}$ where they establish NE during charging transaction.

Our proposed scheme also preserves conditional privacy of the users during electric power transferring and billing process. We do not use any real identity that could lead to the actual user, instead we use a series of legitimate pseudonyms.

Moreover in order to measure the privacy and the anonymity of the vehicles, we calculate the entropy of the user denoted by  $\mathcal{H}$. The anonymity set needed entropy calculation is, the set of active vehicles at the certain time $t$ that are in the process of charging their batteries. Let anonymity set is $U$ and let $p_{U_{i}}$ is the probability that the node $U_{i}$ be the target vehicle whose anonymity is being calculated or $U_{i}$ is under surveillance by adversary $\mathcal{A}$, where $\forall U_{i} \in U$, $\sum_{i=1}^{|U|}p_{U_{i}}=1$.
The entropy $\mathcal{H}$ of the target user $U_i$ in the anonymity set $U$ is given by:
$\mathcal{H}=-\sum_{i=1}^{|U|}{p_{U_{i}}}\times \log_{2}{p_{U_{i}}}$. Since our anonymity set is $U$, the possible outcomes can be $|U|$ assuming the fair distribution and the probability of each outcome will be  $\frac{1}{|U|}$. If the distribution is normal and the occurrence of the nodes to be related to the pseudonyms in question is equally likely, then the maximum entropy is also given by the following formula: $\mathcal{H}_{max}=-\sum_{i=1}^{|U|}p_{U_{i}}\times \log_{2}p_{U_{i}}=\log_{2}p_{U_{i}}$.

\begin{thm}
\textit{In case of any dispute, the node in question can be revoked and the pseudonym in question is linkable to the actual user}
\end{thm}
\begin{proof}
In order to proceed with revocation, RAs get the warrant from the authorities and then look into the $n$ values of the message in question that are provided to RAs in order to figure out which pseudonym was used. After that, RAs collude and construct $x$ from  individual $x_{i}$ related to  the pseudonym in question and the session leader decrypts the keys from  cipher text $c=\{\delta_{1},\delta_{2}\}$ as follows: $PS{_{OBU}^i}=\delta_2\oplus H(x\delta_{1})=(K_{sym}\|K_{V})\oplus H(rPK^{+})\oplus H(rxPK^{+}).$
 When RAs decrypt the keys $K_{sym}$ and $K_{V}$, then revocation is almost done, all RAs have to do is  to decrypt the $(\alpha)_{K_{sym}}$ and then extract ID of the vehicle from the pseudonym.
\end{proof}

\begin{lem}
\textit{It is hard to impersonate other OBU in the process of online power transfer. In other words, it is hard to get away with billing procedures.}
\end{lem}
\begin{proof}
Before starting the power transfer procedure, the vehicles have to register with CSPA in both direct and hash chain-based authentication and provide CSPA with $X_{OBU}$. And at the authentication stage, OBU has to provide the charging plate with $c_{1}$ and $c_{3}$ that contain $H_{2}$ and $H_{3}$ respectively. At the registration phase, $H_{2}$ is associated to the $X_{OBU}$ of the current authenticating vehicle. Therefore any adversary $\mathcal{A}$ with $H_{2}$' without knowing the secret $s$, it will be hard to calculate valid $c_{1}$, $c_{2}$, and $c_{3}$ at the authentication phase.
Therefore the values sent to the CP for authentication will be $c_{1}$', $c_{2}$', $c_{3}$', and $H_{3}$' all of which must have association with the $X_{OBU}$ of the pseudonym $PS{_{OBU}^i}$. Arguing on the collision resistance of the hash function used, it can be inferred that the probability of calculating the right values with not knowing the $X_{OBU}$ is small, therefore it is hard for anybody to impersonate other OBU with a pseudonym.
\end{proof}
The following corollary naturally follows:
\begin{cor}
\textit{Replaying the power transfer request message and/or Pseudonym will not benefit the malicious intent of the user.}

The argument is divided into two parts. Replaying a message will result in the existence of previous power transfer records with this information. Upon successful power transfer, the CSPA maintains a log with timestamp and billing information against the used pseudonym. Let an OBU charges its battery at $CP_{x}$ at particular time $t_{i}$ after successful authentication, the log is recorded at CSPA with the used pseudonym and other credentials. At $t_{i+j}$, the OBU again uses the message, then there are two possibilities. First, the OBU must have already been authenticated before sending this message; in that case, it will receive the power accordingly, secondly if it is not authenticated, then CSPA must have figured out that the record already existed and that the OBU was malicious. In either case, the OBU cannot benefit from such behavior. The same argument stands for the pseudonym as well.
\end{cor}

\subsection{Computation and Communication Overhead}
In this subsection we consider the computation and communication overhead incurred by the OBU and CP in the process of mutual authentication and power transfer. In the computation overhead, we consider the authentication cost incurred by OBU and CP denoted by $T_{auth-OBU}$ and $T_{auth-CP}$ respectively and the cost of revocation denoted by $T_{rev}$ in the direct authentication method. When OBU mutually authenticates with CP, it performs $3H+2EO$ operations, where $H$ denotes the {\bf hash operation} and $EO$ denotes the {\bf exclusive OR} operation. CP performs $6H+5EO$ operations. The cost of revocation $T_{rev}$ in our proposed scheme is given by:\newline
$$T_{rev}=2T_{\gamma}+2T_{mul}+2T_{H}+2T_{dec}$$
$T_{\gamma}$ is the time incurred by the pseudonym search table, $T_{mul}$ is the time required for point multiplication, $T_H$ is the time required to calculate hash, and $T_{dec}$ is the time required for symmetric decryption. In \cite{Zhang2008}, $T_{mul}$ is found for a supersingular curve with embedding $k=6$ over $\mathbb{F}_{3^{97}}$ to be equal to $0.78 ~ms$. Hence the above equations can be written as:
$$
T_{rev}~~=1.56+2(T_{\gamma}+T_{H}+T_{dec})
$$

We also discuss the authentication processing time by both OBU and CP. According to \cite{Chuang2011}, $SHA2$ hash operation takes $0.76 ~\mu sec$. Therefore OBU takes about $2.28 ~\mu sec$ and CP takes about $4.56 ~\mu sec$. It is worth noting that since the XOR operation time is usually a single clock on CPUs which is infinitesimally small, therefore we ignore it.
In case of the hash chain-based authentication, OBU cost is only $1D$, where $D$ denotes the decryption operation. CSPA incurs $1H+1E$, $E$ is the encryption operation.

 We also calculate the authentication overhead. In case of DMA, the communication overhead is equal to $71+u$, where $u$ is the size of the pseudonym. We assume $SHA-512$ as a hash function and consider the \cite{Hari2007} implementation of timestamp which is $6~bytes$. Similarly, in case of PHA the communication overhead is fixed and incurs $135~bytes$ where timestamp is $6~bytes$, $ChargingReq$. is $1~byte$, $X_{OBU}$ is $64~bytes$ and $h_{K_{V}}(\beta)$ is also $64 ~bytes$.

\subsection{Length of Charging Plate}

The efficiency of the power transfer and the auditability depends upon the length of the charging plate. Therefore, the length of the charging plate must be a tradeoff between the authentication and billing delay and the time required to acquire the specified amount of electrical power from the charging plate. The plate should not be too short where a vehicle cannot receive the promised amount of power after spending most of the time on the authentication and billing. Similarly the plate should not be too long, where mutual auditability is on stake and semi-simultaneous audit is not possible. To date, the size of the segment is not fixed; however, OLEV project considers the length of the segment to be $5~m$ which is still controversial because of the authentication and billing overhead incurred by the number of segments. Another important point is that the length of the segment is a design feature where the amount of electrical power, authentication and billing delay, and the time to acquire the guaranteed charge should be taken into account. In our scenario, we argue that the authentication delay incurred by both OBU and charging plate is less than a microsecond (optimistically) due to the design of authentication scheme. Therefore the uniformity of the pickup devices installed in the vehicles, and the capacity of delivering electrical power by the charging plates will play a vital role in deciding the length of the plate.

\subsection{Discussion}
In this subsection we analyze the effect of the two authentication strategies on the efficiency and the design parameters. In DMA, OBU and CP have to perform relatively more operations as compared to PHA; nevertheless the time consumed by these operations (hash and XOR) is less than encryption operation. That is why we argue that in performance, DMA will outperform PHA. Secondly, in DMA, both parties are involved in setting up the session key with mutual agreement. The communication cost is minimum since OBU and charging plate are communicating directly. Therefore the only parameter that could affect the performance of DMA and PHA, is the length of the charging plate. If we consider the normal speed of the vehicle, then PHA will favor the lengthier charging plate than DMA, because of the communication delay incurred by the PHA. On the other hand, PHA does not cost any computation delay because the processing is carried out at resource rich CSPA and charging plate is only used as intermediary. However, in case of PHA, the session key is constructed by one entity, CSPA. Moreover the OBUs must save the hash chain of the currently used pseudonyms in the on-board storage thereby incurring storage cost. Therefore we can argue that, these two methods can be used in different circumstances that fit the necessary conditions for direct and hash-based authentication. For normal scenarios, DMA will be the fair choice because of its security, auditability guarantee, and robustness.

\section{Conclusion}
In this paper, we proposed a secure, privacy-aware, and bidirectional auditable mechanism for wireless power transfer in online electric vehicles. The power transfer technology is installed under the road in the form of charging plates where a segment of the road constitute a charging plate containing a hardware module for communication and lightweight computation. In our proposed scheme, the vehicles use multiple pseudonymous strategy to mutually authenticate with the charging plate and then expedite the power transfer. Meanwhile electric power service provider bills the vehicles on per charging palate basis. Our proposed scheme provides secure and privacy-aware bidirectional auditability where the billing process is verifiable by both parties. Moreover we also present the game theoretic approach to validate the bidirectional auditability.


\section*{Acknowledgment}
This work was supported in part by the NRF (National Research Foundation of Korea) grant funded by the Korea government MEST (Ministry of Education, Science and Technology) (No. NRF-2012R1A2A2A01046986). This research was also supported in part by the MSIP (Ministry of Science, ICT and Future Planning), Korea, under the ITRC (Information Technology Research Center) support program (NIPA-2014-H0301-14-1044) and (NIPA-2014-H0301-14-1015) supervised by the NIPA (National IT Industry Promotion Agency). This work was supported in part by US National Science Foundation (NSF) CREST No. HRD-1345219.




%
\bibliographystyle{IEEETran}
\bibliography{infocom2015}

\end{document}